\let\oldvec\vec
\documentclass{llncs}

\usepackage{llncsdoc}
\usepackage{multirow}
\usepackage{graphicx}
\usepackage{subcaption} 
\let\vec\oldvec
\usepackage{amsmath}
\usepackage{amssymb}
\usepackage{amsfonts}
\usepackage{cite}
\usepackage{cancel}
\usepackage{wasysym}
\usepackage{microtype}
\usepackage{tikz}
\usepackage{makecell}
\usepackage{booktabs}
\usepackage{paralist}
\usepackage{hyperref}
\usetikzlibrary{positioning}
\usetikzlibrary{decorations.pathreplacing}
\usepackage{soul}

\tikzset{
	mybrace/.style={decorate,decoration={brace,aspect=#1}}
}

\newcommand{\N}{\mathbb{N}}

\newcommand{\F}{\mathbb{F}}

\usepackage[usestackEOL]{stackengine}
\begin{document}

\title{On the transversals of Latin squares generated by nonlinear bipermutive cellular automata}

\author{Alberto Dennunzio\inst{1} \and Maximilien Gadouleau\inst{2} \and Luca Mariot\inst{3}}

\institute{$^1$ Department of Informatics, Systems and Communication, University of Milano-Bicocca, Italy \\
            \email{alberto.dennunzio@unimib.it} \\
           $^2$ Department of Computer Science, Durham University, United Kingdom \\
            \email{m.r.gadouleau@durham.ac.uk} \\
           $^3$ Semantics, Cybersecurity and Services Group, University of Twente, The Netherlands \\
           \email{l.mariot@utwente.nl} \\
           }

\authorrunning{A. Dennunzio, M. Gadouleau, L. Mariot}
\titlerunning{On the Transversals of Latin squares generated by nonlinear CA}

\maketitle

\begin{abstract}
In this short paper, we begin to investigate the conditions under which a generic Bipermutive Cellular Automaton (BCA) with no-boundary conditions of diameter $d$ generates a Latin square of order $N=2^{d-1}$ admitting an orthogonal mate, without relying on the linearity of the local rule. Since an orthogonal mate exists if and only if the Latin square can be partitioned into $N$ disjoint \emph{transversals}, we start by characterizing the subclass of BCA whose Latin squares have a transversal on the main diagonal. In particular, we prove that the main diagonal forms a transversal if and only if the generating function of the bipermutive local rule induces an invertible CA with periodic boundary conditions on a configuration of size $d-1$. We then perform exhaustive search experiments, showing that $d=6$ is the smallest diameter for which there exist nonlinear bipermutive CA that generate Latin squares with a transversal on the main diagonal.
\end{abstract}


\section{Introduction}
\label{sec:intro}
Latin squares are combinatorial designs with applications in cryptography~\cite{bruen05}, coding theory~\cite{denes91}, the design of experiments~\cite{hedayat71}, and several other fields. Despite their simple definition, Latin squares to some of the most fascinating questions in combinatorics and discrete mathematics\cite{denes-ls}.

A recent line of works considered how to construct Latin squares via Cellular Automata (CA) (see e.g.~\cite{manzoni26} for a recent survey of relevant results). In particular, the authors of~\cite{mariot20} showed that a no-boundary CA of length $2(d-1)$ equipped with a bipermutive local rule (i.e., a Bipermutive CA, or BCA) of diameter $d$ generates a Latin square of order $N = 2^{d-1}$. Moreover, that work also investigated under which conditions the Latin squares generated by two BCA are \emph{orthogonal}, that is, their superposition yields all the ordered pairs in the Cartesian product $[N]\times [N]$, where $[N] = \{1,\ldots, N\}$. More precisely, by focusing on the subclass of Linear BCA (LBCA), the authors of~\cite{mariot20} proved that the orthogonality of the Latin squares is equivalent to the coprimality of the polynomials associated to the underlying linear local rules. The motivation to study this type of constructions is that families of Mutually Orthogonal Latin Squares (MOLS) can be used to design threshold secret sharing schemes and authentication codes, among other cryptographic primitives~\cite{stinson-cd}. 

The rich body of literature leveraging the construction of orthogonal Latin squares through LBCA~\cite{mariot22,mariot23,Hammer23,gadouleau23} contrasts with what little is known about Latin squares induced by \emph{nonlinear} BCA. So far, only a necessary condition on the local rules of two nonlinear BCA has been exhibited in~\cite{mariot17a} to generate orthogonal Latin squares, while~\cite{mariot17b} proposed a heuristic approach based on evolutionary algorithms to construct pairs of nonlinear BCA that generate orthogonal Latin squares. A complete characterization of nonlinear BCA-based orthogonal Latin squares is still a wide open problem.

In this exploratory paper, we start to tackle the study of Latin squares generated by nonlinear BCA from the perspective of their \emph{transversals}. The reason is that a Latin square admits $N$ disjoint transversals (with $N$ being the order of the square) if and only if it has an \emph{orthogonal mate}, i.e. another Latin square orthogonal to it~\cite{denes-ls}. In particular, in the context of this paper we focus on a simpler problem, namely studying under which conditions a nonlinear BCA-based Latin square has a transversal on its main diagonal. Our main result is a characterization of such squares in terms of the invertibility of a periodic-boundary CA that defines the main diagonal of the square, equipped with the generating function of the BCA as a local rule. Finally, we perform a computational search for such Latin squares generated by BCA of small diameters, observing that our characterization needs at least diameter $d=6$ to yield nonlinear rules.

\section{Preliminaries}
\label{sec:bg}

\subsection{Latin Squares}
\label{subsec:ls}
Here, we only recall the essential definitions related to Latin squares. The reader is referred to~\cite{denes-ls} for a more complete overview of this subject.

\begin{definition}
\label{def:ls}
A Latin square of order $N \in \N$ is an $N \times N$ matrix $L$ with entries from $[N]$ such that for all $i, j, k \in [N]$ with $j\neq k$, one has that $L(i, j) \neq L(i, k)$ and $L(j, i) \neq L(k, i)$.
\end{definition}
Thus, in a Latin square each number from $1$ to $N$ occurs exactly once in each row and in each column of a Latin square. Alternatively, each row and each column of $L$ forms a permutation of $[N]$. The next definition formalizes the notion of orthogonal Latin squares:
\begin{definition}
\label{def:ols}
Two Latin squares $L_1, L_2$ of order $N \in \N$ are orthogonal if for all distinct pairs of coordinates $(i, j), (i', j') \in [N] \times [N]$ it holds that
\begin{equation}
(L_1(i,j), L_2(i,j)) \neq (L_1(i', j'), L_2(i', j')) \enspace .
\end{equation}
\end{definition}
Hence, two Latin squares of order $N$ are orthogonal if their \emph{superposition} yields all ordered pairs of the Cartesian product $[N] \times [N]$ exactly once.

A \emph{transversal} in a Latin square $L$ of order $N$ is a set $T$ of $N$ coordinates pair $T = \{(i_1,j_1), \ldots, (i_N,j_N)\}$ where for each distinct pair $(i_k, j_k), (i_l, j_l) \in T$ it holds that $i_k \neq i_l$ and $j_k \neq j_l$, and $\{L(i_1, j_1), \ldots, L(i_N, j_n)\}$ is a permutation of $[N]$. Stated otherwise, each number from $1$ to $N$ occurs only once in the entries of $L$ defined by the transversal $T$. It is well-known that a Latin square $L$ admits an \emph{orthogonal mate} (i.e. another Latin square $M$ that is orthogonal to $M$) if and only if it can be decomposed into $N$ disjoint transversals~\cite{denes-ls}.

\subsection{Cellular Automata}
\label{subsec:ca}
A Cellular Automaton (CA) is a computational model characterized by a regular lattice of \emph{cells}, each of which synchronously updates its state by applying a \emph{local rule} on the configuration defined by itself and its neighboring cells. In this paper we consider \emph{No-Boundary CA} (NBCA) and \emph{Periodic-Boundary CA} (PBCA), which we formally introduce below. Further, we consider only \emph{binary} CA, meaning that each cell can only be in a state $s \in \F_2 = \{0,1\}$.
\begin{definition}
	\label{def:ca}
	Let $d, n \in \N$ be positive integers such that $d\le n$. Further, let $f: \F_2^d \to \F_2$ be a $d$-variable Boolean function. Then, we define the following two models
	of \emph{one-dimensional binary} CA with $n$ input cells, equipped with the \emph{local rule} $f$ of \emph{diameter} $d$:
	\begin{itemize}
		\item \emph{No Boundary CA} (NBCA): $F: \F_2^{n} \rightarrow \F_2^{n-d+1}$ is
		defined for all $x \in \F_2^n$ as:
		\begin{equation}
		\label{eq:nbca}
		F(x_1, \ldots, x_n) = (f(x_1, \ldots, x_d), f(x_2, \ldots, x_{d+1}),
		\ldots, f(x_{n-d+1}, \ldots, x_n)) \enspace .
		\end{equation}
		\item \emph{Periodic Boundary CA} (PBCA): $F: \F_2^{n} \rightarrow \F_2^{n}$ is defined for all $x \in \F_2^n$ as:
		\begin{equation}
		\label{eq:pbca}
		F(x_1, \ldots, x_n) = (f(x_1, \ldots, x_d), \ldots,
		f(x_{n-d},\ldots,x_1), \ldots, f(x_n,\ldots, x_{d-1}) \enspace .
		\end{equation}
	\end{itemize}  
\end{definition}
Thus, the difference between a NBCA and a PBCA lies in how we deal with the cells at the boundaries: in the former, the local rule is only applied to the cells that have enough neighbors on the right, while in a PBCA it is applied to all cells by considering the cell indices as integers modulo $n$.
From a representation standpoint, the truth table of $f$ can be uniquely identified by the $2^d$-bit string $\Omega_f \in \F_2^{2^d}$ of its outputs, once we fix a total order on the input vectors of $\F_2^d$ (e.g., the lexicographic order). The \emph{Wolfram code} of the local rule $f$, is the decimal encoding of $\Omega_f$~\cite{wolfram83}.

A local rule $f: \F_2^d \to \F_2$ is \emph{bipermutive} if there exists $g: \F_2^{d-2} \to \F_2$ such that $f$ can be written for all $(x_1,x_2, \ldots, x_{d-1}, x_d) \in \F_2^d$ as~\cite{leporati14}:
$$
f(x_1,x_2, \ldots, x_{d-1}, x_d) = x_1 \oplus g(x_2, \ldots, x_{d-2}) \oplus x_d \enspace .
$$
Thus, a bipermutive rule depends linearly on the leftmost and rightmost variables.

Finally, the Lemma below briefly recalls how NBCA equipped with bipermutive local rules generate Latin squares, a fact that has been re-discovered independently in different works~\cite{eloranta93,mariot16}. In what follows, we assume that $\F_2^{d-1}$ is lexicographically ordered, and that $\phi: \F_2^{d-1} \rightarrow [N]$ is a monotone bijective mapping from $\F_2^{d-1}$ to $[N]$, with $\psi = \phi^{-1}$ denoting the inverse mapping.
\begin{lemma}
\label{lm:lat-sq-bip-ca}
Let $F: \F_2^{2(d-1)} \rightarrow \F_2^{d-1}$ be a NBCA defined by a $d$-diameter local rule $f: \F_2^d \rightarrow \F_2$, and let $N = 2^{d-1}$. Define the square associated to $F$ as the $N \times N$ matrix $S_F$ with entries from $[N]$ such that
    \begin{equation}
        \label{eq:sq-ca}
        S_{F}(i,j) = \phi(F(\psi(i)||\psi(j))) \enspace ,
    \end{equation}
    for all $1 \le i,j \le N$, where $||$ denotes concatenation. Then, $S_F$ is a Latin square of order $N$ if and only if $f$ is bipermutive.
\end{lemma}
From an intuitive point of view, the output of the CA is taken as the entry of the square at the row and column coordinates encoded by the left and the right $(d-1)$-bit blocks of the CA input state.

\section{Transversal on the Main Diagonal}
\label{sec:main}
As recalled in the previous section, the existence of $N$ disjoint transversals in a Latin square $L$ is equivalent to the existence of an orthogonal mate of $L$. However, characterizing the conditions under which a BCA generates a Latin square with $N$ disjoint transversals seems like a tough combinatorial problem, which is beyond the scope of an exploratory paper like the present one. Thus, we start with an easier version of the CA transversal problem: namely, \emph{when does a BCA generate a Latin square that features a transversal on the main diagonal}? The formal statement of the problem is given below:
\begin{problem}
\label{pb:statement}
Let $F: \F_2^{2(d-1)} \to \F_2^{d-1}$ be a NBCA equipped with a bipermutive local rule $f: \F_2^d \to \F_2$ of diameter $d$, where 
$$
f(x_1,x_2, \ldots, x_{d-1}, x_d) = x_1 \oplus g(x_2, \ldots, x_{d-2}) \oplus x_d \enspace ,
$$
with $g: \F_2^{d-2} \to \F_2$. Further, let $T: \F_2^{d-1} \to \F_2^{d-1}$ be the mapping defined for all $x \in \F_2^{d-1}$ as $T(x) = F(x||x)$, i.e. the function whose image corresponds to the main diagonal of the Latin square $S_F$ generated by $F$. What conditions must $f$ satisfy to ensure that $T$ is a permutation, or equivalently that the main diagonal of $S_F$ is a transversal?
\end{problem}

Notice that we do not impose any linearity constraint on the local rule of the CA. In particular, linear CA already have a rich algebraic theory one can leverage to show that they all admit $N$ disjoint transversals, since they always have an orthogonal mate. This is due to the fact that, as proven in~\cite{mariot20}, the orthogonality of two linear BCA is equivalent to the coprimality of the polynomials associated to their local rules. Starting from a single linear BCA, it is easy to show that its associated polynomial always admits a coprime pair. Therefore, in this paper we focus on the generic case where the BCA rule can also be nonlinear, which is the most interesting one.

Let us now take a closer look at the mapping $T$. This corresponds to the application of the CA global rule $F$ on a very specific input state, namely the concatenation of a block $(x_1,\ldots, x_{d-1}) \in \F_2^{d-1}$ with itself, as depicted in Figure~\ref{fig:symm}.
\begin{figure}[t]
    \centering
    \begin{tikzpicture}
        [->,auto,node distance=1.5cm, empt node/.style={font=\sffamily,inner
            sep=0pt}, rect
        node/.style={rectangle,draw,font=\bfseries,minimum size=0.75cm, inner
            sep=0pt, outer sep=0pt}]
        
        \node [empt node] (c)   {};
        \node [rect node] (c1) [right=0.1cm of c] {$y_1$};
        \node [rect node] (c2) [right=0cm of c1] {$y_2$};
        \node [rect node] (c3) [right=0cm of c2, minimum width=2.25cm] {$\cdots$};
        \node [rect node] (c4) [right=0cm of c3] {$y_{d-1}$};
        
        \node [empt node] (f1) [above=0.2cm of c3.north] {\Large $\Downarrow$ \normalsize $F$};
        
        \node [rect node] (p3) [above=0.85cm of c1.north, minimum width=2.25cm] {$\cdots$};
        \node [rect node] (p2) [left=0cm of p3] {$x_2$};
        \node [rect node] (p1) [left=0cm of p2] {$x_1$};
        \node [rect node] (p4) [right=0cm of p3] {$x_{d-1}$};
        \node [rect node] (p5) [right=0cm of p4] {$x_1$};
        \node [rect node] (p6) [right=0cm of p5] {$x_2$};
        \node [rect node] (p7) [right=0cm of p6, minimum width=2.25cm] {$\cdots$};
        \node [rect node] (p8) [right=0cm of p7] {$x_{d-1}$};

        \draw[-,very thick] (p1.north west) -- (p8.north east);
        \draw[-,very thick] (p1.south west) -- (p8.south east);
        \draw[-,very thick] (p1.south west) -- (p1.north west);
        \draw[-,very thick] (p8.south east) -- (p8.north east);
        \draw[-,very thick] (p4.south east) -- (p4.north east);

        \draw[-,very thick] (c1.south west) -- (c4.south east);
        \draw[-,very thick] (c1.south west) -- (c1.north west);
        \draw[-,very thick] (c1.north west) -- (c4.north east);
        \draw[-,very thick] (c4.south east) -- (c4.north east);

        \draw [-, mybrace=0.5, decorate, decoration={brace,amplitude=10pt,raise=0.1cm}]
                (p1.north west) -- (p4.north east) node [midway,yshift=0.4cm] {$x \in \F_2^{d-1}$};
        \draw [-, mybrace=0.5, decorate, decoration={brace,amplitude=10pt,raise=0.1cm}]
                (p5.north west) -- (p8.north east) node [midway,yshift=0.4cm] {$x \in \F_2^{d-1}$};
        \draw [-, mybrace=0.5, decorate, decoration={brace,mirror, amplitude=10pt,raise=0.1cm}]
                (c1.south west) -- (c4.south east) node [midway,yshift=-1cm] {$y = F(x||x) = T(x)$};

    \end{tikzpicture}
    \caption{Application of the global rule $F$ to the symmetric configuration $x||x$, with $x \in \F_2^{d-1}$. This is equivalent to computing $T(x)$.}
    \label{fig:symm}
\end{figure}
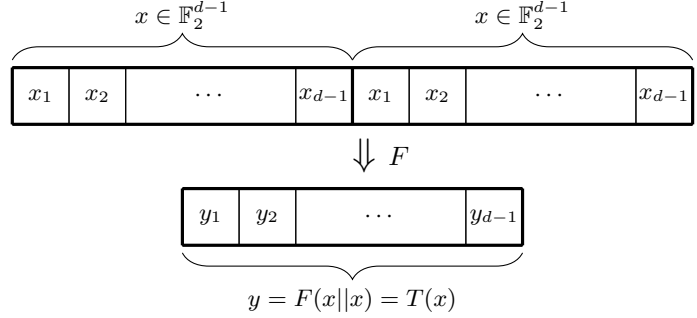
On the square representation, this means that we are considering only the coordinates where the row and column coordinates coincide, or equivalently the main diagonal of the square $S_F$. The next remark gives an algebraic characterization of the output of $T$ in terms of the generating function of the local rule.
\begin{remark}
\label{rem:alg-ex}
Let $T: \F_2^{d-1} \to \F_2^{d-1}$ be defined as in~\ref{pb:statement}. Then, for all input vectors $x = (x_1,\ldots, x_{d-1}) \in \F_2^{d-1}$, the output $y = T(x) = F(x||x)$ is computed as $y = (y_1,y_2, \ldots, y_{d-1}) $, where:
\begin{align*}
    \label{eq:output}
        y_1 &= f(x_1, x_2, \ldots, x_{d-1}, x_1) = \cancel{x_1} \oplus g(x_2,\ldots, x_{d-1}) \oplus \cancel{x_1} = g(x_2,\ldots, x_{d-1}) \\
        y_2 &= f(x_2, x_3, \ldots, x_1, x_2) = \cancel{x_2} \oplus g(x_3,\ldots, x_1) \oplus \cancel{x_2} = g(x_3,\ldots, x_1) \\
        &\vdots \\
        y_{d-1} &= f(x_{d-1}, x_1, \ldots, x_{d-2}, x_{d-1}) = \cancel{x_{d-1}} \oplus g(x_1,\ldots, x_{d-2}) \oplus \cancel{x_{d-1}} = g(x_1,\ldots, x_{d-2}) \\
\end{align*}
In other words, the computation of $T$ \emph{only depends on the generating function $g$ computed on the central cells of the neighborhood}, since the leftmost and rightmost variables are always equal, and thus they cancel out under the XOR operation.
\end{remark}

More interestingly, the next remark shows that $T$ actually corresponds to a CA with periodic boundary conditions: 

\begin{remark}
\label{rem:pbca}
The input vectors $(x_2, \ldots, x_{d-1})$, $(x_3, \ldots, x_1)$, $\ldots$, $(x_1, \ldots, x_{d-2})$ upon which $g$ is evaluated corresponds to the cyclic shifts over the configuration $(x_2, x_3, \ldots, x_{d-1}, x_1)$, which has size $d-1$. Thus, the map $T: \F_2^{d-1} \to \F_2^{d-1}$ is actually defined by a CA with periodic boundary conditions having $d-1$ cells, and equipped with the generating function $g: \F_2^{d-2} \to \F_2$ as a local rule. Figure~\ref{fig:pbca} displays $T$ as a PBCA.
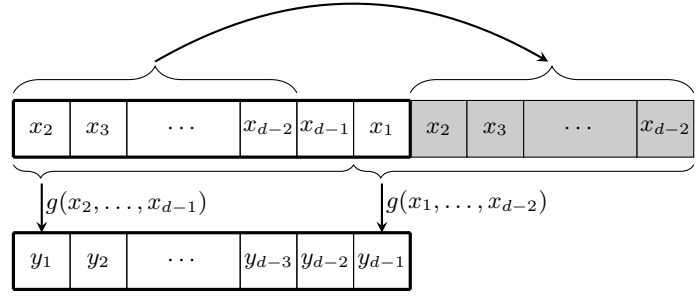
\begin{figure}[t]
\centering
\begin{tikzpicture}
            [->,auto,node distance=1.5cm, empt node/.style={font=\sffamily,inner
                sep=0pt}, rect
            node/.style={rectangle,draw,font=\sffamily\bfseries,minimum size=0.75cm, inner
                sep=0pt, outer sep=0pt}, grey node/.style={rectangle,draw,fill=gray!40,
                font=\sffamily\bfseries,minimum size=0.75cm, inner sep=0pt, outer sep=0pt}]
            
            \node [empt node] (c)   {};
            \node [rect node] (c1) [right=0.1cm of c] {$y_2$};
            \node [rect node] (c0) [left=0cm of c1] {$y_1$};
            \node [rect node] (c2) [right=0cm of c1, minimum width=1.5cm] {$\cdots$};
            \node [rect node] (c3) [right=0cm of c2] {$y_{d-3}$};
            \node [rect node] (c4) [right=0cm of c3] {$y_{d-2}$};
            \node [rect node] (c5) [right=0cm of c4] {$y_{d-1}$};
            
            \node [rect node] (p2) [above=1cm of c1] {$x_3$};
            \node [rect node] (p1) [left=0cm of p2] {$x_2$};
            \node [empt node] (p)  [left=0.1cm of p1] {};
            \node [rect node] (p3) [right=0cm of p2, minimum width = 1.5cm] {$\cdots$};
            \node [rect node] (p4) [right=0cm of p3] {$x_{d-2}$};
            \node [rect node] (p5) [right=0cm of p4] {$x_{d-1}$};
            \node [rect node] (p6) [right=0cm of p5] {$x_1$};
            \node [grey node] (p7) [right=0cm of p6] {$x_2$};
            \node [grey node] (p8) [right=0cm of p7] {$x_3$};
            \node [grey node] (p9) [right=0cm of p8, minimum width = 1.5cm] {$\cdots$};
            \node [grey node] (p10) [right=0cm of p9] {$x_{d-2}$};
            \node [empt node] (p11) [above=0.85cm of p3.west] {};
            \node [empt node] (p12) [right=0.3cm of p11] {};
            \node [empt node] (p13) [above=0.85cm of p9.west] {};
            \node [empt node] (p14) [right=0.3cm of p13] {};
            
            \draw [-, mybrace=0.5, decorate, decoration={brace,amplitude=10pt,raise=0.1cm}]
                (p1.north west) -- (p4.north east);
            \draw [-, mybrace=0.5, decorate, decoration={brace,amplitude=10pt,raise=0.1cm}]
                (p7.north west) -- (p10.north east);
            \draw [->, thick, shorten >=0pt,shorten <=0pt,>=stealth] (p12) edge[bend left] (p14);

            \draw [-, mybrace=0.08, decorate, decoration={brace, mirror, amplitude=5pt,raise=0.1cm}] (p1.south west) -- (p5.south east);
            \node [empt node] (e1) [below=0.3cm of p1.south] {};
            \draw [->, thick, shorten >=0pt,shorten <=0pt,>=stealth] (e1) -- (c0.north);
            \node [empt node] (e2) [above=0.6cm of c1.east] {$g(x_2,\ldots,x_{d-1})$};

            \draw [-, mybrace=0.08, decorate, decoration={brace, mirror, amplitude=5pt,raise=0.1cm}] (p6.south west) -- (p10.south east);
            \node [empt node] (e3) [below=0.3cm of p6.south] {};
            \draw [->, thick, shorten >=0pt,shorten <=0pt,>=stealth] (e3) -- (c5.north);
            \node [empt node] (e4) [below=0.8cm of p7.east] {$g(x_1,\ldots,x_{d-2})$};

            \draw[-,very thick] (p1.north west) -- (p6.north east);
            \draw[-,very thick] (p1.south west) -- (p6.south east);
            \draw[-,very thick] (p1.north west) -- (p1.south west);
            \draw[-,very thick] (p6.north east) -- (p6.south east);

            \draw[-,very thick] (c0.north west) -- (c5.north east);
            \draw[-,very thick] (c0.south west) -- (c5.south east);
            \draw[-,very thick] (c0.north west) -- (c0.south west);
            \draw[-,very thick] (c5.north east) -- (c5.south east);
                
        \end{tikzpicture}
\caption{Depiction of the map $T$ as a PBCA of size $d-1$, with local rule corresponding to the generating function $g: \F_2^{d-2} \to \F_2$.}
\label{fig:pbca}
\end{figure}
\end{remark}

We can now prove our main result: the Latin square generated by a NBCA defined by a bipermutive local rule $f$ has a transversal on the main diagonal if and only if the underlying generating function induces an invertible PBCA.

\begin{theorem}
\label{thm:main}
Let $F: \F_2^{d-1} \to \F_2^{d-1}$ be a NBCA equipped with a bipermutive local rule $f: \F_2^{d} \to \F_2$ of diameter $d$, where $f$ is defined by a generating function $g: \F_2^{d-2} \to \F_2$ for all $x = (x_1, x_2, \ldots, x_{d-1}, x_d)$ as 
\begin{equation}
\label{eq:bip}
f(x_1,x_2, \ldots, x_{d-1}, x_d) = x_1 \oplus g(x_2, \ldots, x_{d-1}) \oplus x_d \enspace .
\end{equation}
Then, the Latin square $S_F$ of order $N=2^{d-1}$ generated by $F$ has a transversal on the main diagonal if and only if the generating function $g$, when considered as a local rule of diameter $d-2$, induces an invertible PBCA $T: \F_2^{d-1} \to \F_2^{d-1}$.

\begin{proof}
Suppose that $S_F$ has a transversal on its main diagonal. Then, it follows that the NBCA $F$ evaluated on all symmetric $2(d-1)$-bit blocks $x||x$ is a permutation of $F_2^{d-1}$. By Remark~\ref{rem:alg-ex}, the $d-1$ output coordinates of this permutation $T(x) \equiv F(x||x)$ only depends on the generating function $g$, and by Remark~\ref{rem:pbca}, the map $T$ corresponds to the PBCA equipped with rule $g$ evaluated on the vector $(x_2, \ldots, x_{d-1}, x_1)$ of size $d-1$. Since the PBCA $T$ is shift invariant, without loss of generality one can also consider $T$ as a PBCA defined on the vector $(x_1, x_2, \ldots, x_{d-1})$, which is simply a cyclic rotation of $(x_2, \ldots, x_{d-1}, x_1)$.

Conversely, suppose that the generating function $g: \F_2^{d-2} \to \F_2$ induces a PBCA $T: \F_2^{d-1} \to \F_2^{d-1}$ that is invertible, namely a shift-invariant bijective mapping $(x_1, x_2, \ldots, x_{d-1}) \mapsto (y_1, y_2, \ldots, y_{d-1})$. Consider now the the NBCA $F: \F_2^{2(d-1)} \to \F_2^{d-1}$ defined by the bipermutive local rule $f$ of Eq.~\ref{eq:bip}, with $S_F$ being the Latin square of order $2^{d-1}$ generated by $F$. By Remarks~\ref{rem:alg-ex} and~\ref{rem:pbca}, the diagonal of $S_F$ corresponds to $F(x||x) = T(X)$ for all $x \in \F_2^{d-1}$, just applied on a cyclically shifted vector $(x_2, \ldots, x_{d-1}, x_1)$. Since $F$ is also shift-invariant, then one can consider $T$ as defined on the cyclic input vector $(x_1, \ldots, x_{d-1})$. \qed
\end{proof}
\end{theorem}

\section{Computational Search Experiments}
\label{sec:comp}
On account of Theorem~\ref{thm:main}, checking if a BCA generates a Latin square with a transversal on the main diagonal can be reduced to verifying that the PBCA $T: \F_2^{d-1} \to \F_2^{d-1}$ induced by the generating function $g: \F_d^{d-2} \to \F_2^{d-1}$ of the bipermutive rule is invertible. Notice that invertibility does not necessarily mean that the PBCA $T$ must also be \emph{reversible}: what we are looking for are the so-called \emph{globally invertible rules} i.e., rules inducing invertible PBCA only for certain array sizes. A famous example of globally invertible rule is the transformation $\chi: \F_2^3 \to \F_2$ defined as $f(x_1, x_2, x_3) = x_1 \oplus x_2 \oplus x_2x_3$, which give an invertible CA whenever the array size is odd~\cite{daemen94}. However, $\chi$ is not a good example for our construction: taken as a generating function of $d-2=3$ variables, it does not yield an invertible PBCA of length $d-1=4$ cells, since the length of the array is even. Other examples of these \emph{complementing landscapes rules}\footnote{The name comes from the fact that the update rule can be expressed as flipping the state of the cell if and only if the neighborhood is in a particular landscape, which is a specific form of regular expression.} can be found in the appendix of Daemen's PhD thesis~\cite{daemen95}, as well as in~\cite{mariot19,haugland24}.

To make our investigation more systematic, we implemented an exhaustive search algorithm that, given in input the diameter $d$ of the bipermutive local rule, search among all the $2^{2^{d-2}}$ generating functions of $d-2$ variables those that define invertible PBCA of length $d-1$. In particular, we performed our exhaustive search experiments up to $d=6$, which means looking at $65\, 536$ generating functions in the largerst considered instance. Interestingly, up to diameter $d=5$ our exhaustive search \emph{only found linear rules} that generated invertible PBCA of length $d-1$, or equivalently Latin squares with a transversal on the main diagonal. On the other hand, for $d=6$ our search yielded $472$ rules in total, out of which $456$ were nonlinear. As an example, Figure~\ref{fig:ex-diag} in the Appendix reports the $32\times 32$ Latin square generated by the BCA $F: \F_2^{10} \to \F_2^5$ whose bipermutive rule is defined by the nonlinear generating function $g(x_1,x_2,x_3,x_4)=x_1 \oplus x_3 \oplus x_1x_4$.

\section{Conclusions}
\label{sec:outro}
Clearly, the road to a complete characterization of all transversals in BCA-based Latin squares is still quite long: so far, we managed to find a class of rules that induce a transversal on the main diagonal. Then, the next natural step of this line of investigation is to extend the result to other transversals, with the goal of constructing other $N-1$ disjoint ones beside the main diagonal. In turn, this would yield a construction for nonlinear BCA that generate Latin squares with an orthogonal mate. A interesting direction to tackle this problem is to define other transversals in terms of the main diagonal. For example, one could think of applying a linear transformation (such as an XOR shift) on the coordinates of the diagonal, and then check if the new set of coordinates yields a permutation. We plan to investigate this idea in future research.

\bibliographystyle{splncs04}
\bibliography{bibliography}

\section*{Appendix: Example of nonlinear BCA-based Latin square}

\begin{figure}[!]
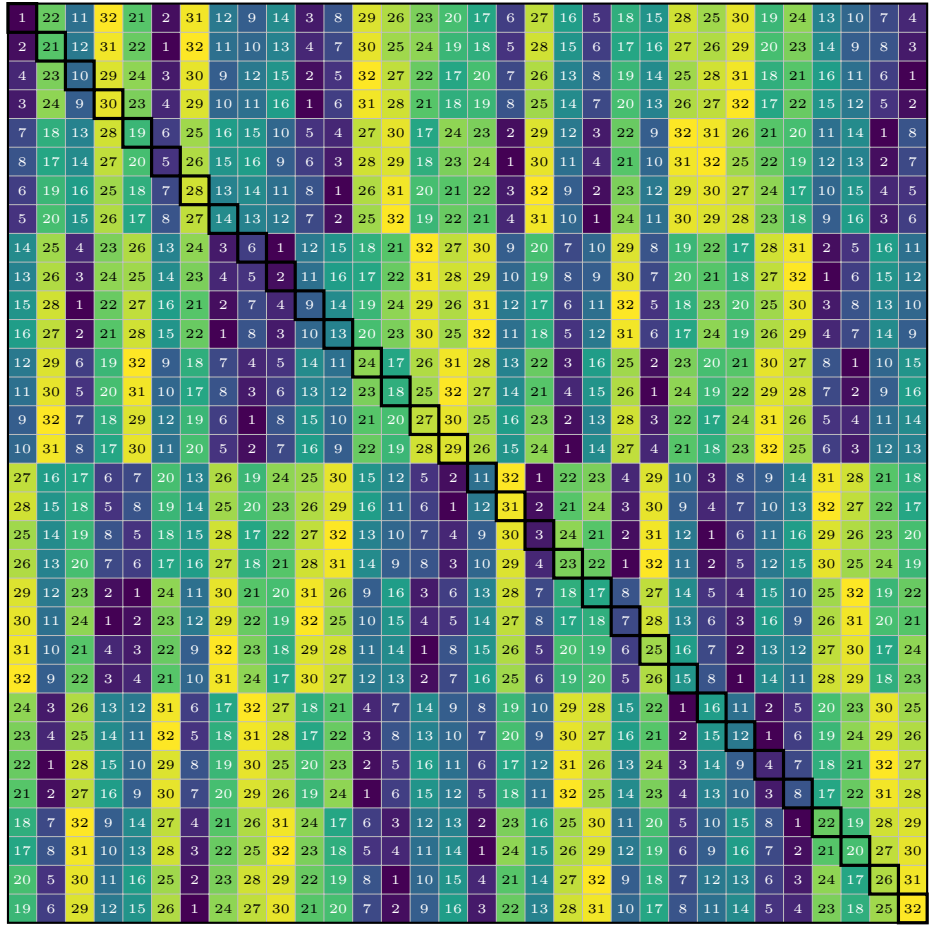

\centering

\caption{The $32\times32$ Latin square generated by the BCA of diameter~6 with generating function $g(x_1,x_2,x_3,x_4)=x_1\oplus x_3\oplus x_1x_4$. Background shading follows the viridis colormap (symbol~1: dark purple, symbol~32: yellow). Cells outlined in black form the main diagonal transversal.}
\label{fig:ex-diag}
\end{figure}

\end{document}